\documentclass[11pt]{article}
\usepackage[margin=1in]{geometry} 
\usepackage{fullpage}
\usepackage{graphicx}
\usepackage{tikz}
\usetikzlibrary{decorations.pathreplacing,calc}

\usepackage{array}
\usepackage{amsthm}
\usepackage{amsmath}
\usepackage{amssymb}
\usepackage{hyperref}
\usepackage{svg}
\usepackage{cleveref}
\usepackage[normalem]{ulem}

\newcommand{\T}{{\mathcal T}}
\newcommand{\E}{{\mathcal E}}

\newcommand{\tl}{\mathrm{tl}}

\renewcommand{\part}{\mathrm{part}}
\newcommand{\anc}{\mathrm{anc}}
\newcommand{\comp}{\mathrm{comp}}

\usepackage{xcolor}

\newcommand{\lemref}[1]{Lemma~\ref{#1}}

\newtheorem{lemma}{Lemma}
\newtheorem{theorem}{Theorem}
\newtheorem{corollary}{Corollary}

\title{
    Reconstructing Bounded Treelength Graphs with Linearithmic Shortest Path Distance Queries
}

\author{
    Chirag Kaudan\thanks{Oregon State University, \texttt{kaudanc@oregonstate.edu}} \and
    Amir Nayyeri\thanks{Oregon State University, \texttt{amir.nayyeri@oregonstate.edu}}
    \thanks{
        The authors were supported by NSF grants CCF-1941086 and CCF-2311180.
    }  
}

\index{Author, First}
\index{Researcher, Second}

\begin{document}
\thispagestyle{empty}
\maketitle

\begin{abstract}

We consider the following graph reconstruction problem: given an unweighted connected graph $G = (V,E)$ with visible vertex set $V$ and an oracle which takes two vertices $u,v \in V$ and returns the shortest path distance between $u$ and $v$, how many queries are needed to reconstruct $E$?

Specifically, we consider bounded degree $\Delta$ and bounded treelength $\mathrm{tl}$ connected graphs and show that reconstruction can be done in $O_{\Delta,\mathrm{tl}}(n \log n)$ queries with a deterministic algorithm. This result improves over the best known algorithm (deterministic or randomized) for this graph class by a $\log n$ factor and matches the known lower bound for the class of graphs with bounded chordality, which is a subclass of bounded treelength graphs.

\end{abstract}

\section{Introduction}
In the graph reconstruction problem, only the vertex set \( V \) of a hidden graph \( G = (V, E) \) is visible, and the goal is to reconstruct the edge set \( E \) by querying an oracle. A graph reconstruction algorithm aims to discover the edge set with as few adaptively chosen queries as possible. The reconstruction task is considered complete once there is a unique graph consistent with all answers to the queries. Several query models have been studied for graph reconstruction, including edge existence queries, edge counting, effective resistance queries, and shortest path queries~\cite{Reyzin2007, AngluinChen2004, HMMT2018, ABKRS2004, AlonAsodi2005, Bennett2025Graphinference}.

In the Shortest Path (SP) distance query model, we have access to a distance oracle: for any two vertices \( u, v \in V \), a query returns the length of the shortest path between \( u \) and \( v \) in \( G \) and returns $\infty$ if no such path exists. 

The study of network reconstruction using shortest path queries is motivated by applications such as Internet topology reconstruction and evolutionary tree reconstruction~\cite{Hein1989, WSSB1977, Beerliova2005, Dallasta2006, ACKM2009}.

Any (unweighted) graph can trivially be reconstructed using \( \binom{n}{2} \) SP queries, simply by querying every pair of vertices to check whether they are adjacent. Unfortunately, this upper bound is tight even for trees, once we allow unbounded degree~\cite{Reyzin2007}.

Moreover, the bound is tight if one allows disconnected graphs--reconstructing a graph with $n$ vertices and $1$ edge requires $\Omega(n^2)$ queries.
In light of these observations, all non-trivial algorithmic results on graph reconstruction focus on connected bounded degree graphs.

Mathieu et al.~\cite{MathieuZhou2013} \cite{KMZ2018} showed that general bounded-degree graphs can be reconstructed using \( \widetilde{O}_{\Delta}(n^{3/2}) \) expected queries via a randomized algorithm.\footnote{We use $\widetilde{O}(f(n))$ as a shorthand for $O(f(n)\cdot \text{polylog}(n))$.  Also, we use $O_p(\cdot)$ for different paratmeters $p$ to indicate that the hidden constant depends on $p$.}. This is still the best known general upper bound today. It is unknown whether there exists a deterministic algorithm to reconstruct a bounded-degree graph in \(\widetilde{O}_{\Delta}(n^{3/2})\) queries. Mathieu et al.~further described a randomized algorithm that reconstructs chordal graphs using only \( O_\Delta(n \log^3 n) \) queries. Building on this, Bastide and Groenland~\cite{Bastide2024a} strengthened their result by obtaining a reconstruction algorithm for bounded degree \( k \)-chordal graphs with $O_{\Delta,k}(n\log n)$ queries (See also Rong et al.~\cite{Rong2021} for other nearly linear reconstruction algorithms of nearly chordal graphs).\footnote{A graph is \( k \)-chordal if all its induced cycles have length at most \( k \); chordal graphs are \( 3 \)-chordal.}

In a more recent work, Bastide and Groenland~\cite{Bastide2024b} describe a randomized \( O_{\Delta, \tau}(n \log^2 n) \)-query algorithm for reconstructing graphs with treelength at most $\tau$. Note that bounded chordality implies bounded treelength~\cite{journals/dm/DourisboureG07}, but the converse is not true. Thus, their latter result applies to a broader class of graphs, though the algorithm is slightly slower and randomized. 

In this paper, we present a deterministic \( O_{\Delta,\tau}(n \log n) \)-query algorithm for reconstructing graphs of bounded degree $\Delta$ and treelength at most $\tau$--our algorithm is deterministic and more efficient by a factor of $\log n$ compared to Bastide and Groenland~\cite{Bastide2024b}.

\begin{theorem}
\label{thm:reconstruction_algo}
Given a hidden graph \( G \), and an integer $\tau\geq \tl(G)$, there is an algorithm that reconstructs \( G \) using at most 
\(
O(\Delta^{3\tau+2} \cdot n \log n)
\)
SP queries, where $n$ is the number of vertices in $G$ and $\Delta$ is its maximum degree.
\end{theorem}

In particular, our result implies a deterministic \( O_{\Delta,k}(n \log n) \)-query algorithm for reconstructing bounded degree graphs with bounded chordality, matching the optimal result for reconstructing graphs with bounded $k$-chordality~\cite{Bastide2024a}. Note that this theorem assumes that the parameter $\tau$ is given.

\section{Background and notation}
In this paper, all graphs are simple, connected, unweighted and undirected. For a graph $G = (V,E)$ and vertices $u,v \in V(G)$, we use $d_G(u,v)$ to denote the length of the shortest path between $u$ and $v$ in $G$. For a graph $G = (V,E)$ and a vertex $v \in V(G)$, we use $N_G(v)$ to denote the set of all neighbors of $v$ in the graph $G$ and $N_G[v] = N_G(v) \cup \{v\}$, the open and closed neighborhood of $v$ respectively. Analogously, for a graph $G = (V,E)$ and a set $S \subseteq V(G)$, we let $N_G[S] = \{u \in V(G) \mid \exists a \in S, d_G(a,u) \leq  1\}$ and $N_G(S) = N_G[S] \setminus S$ be the closed neighborhood and open neighborhood of $S$ in $G$ respectively. For a graph $G = (V,E)$ and a path $\pi = v_0v_1\dots v_k$ in $G$, we use $\pi[v_i,v_j]$ to denote the subpath of $\pi$ starting at $v_i$ and ending at $v_j$ for $0 \leq i \leq j \leq k$. Further, we denote the concatenation of paths $\pi = \pi_0, \dots, \pi_n$ and $\rho = \rho_0, \dots, \rho_m$ where $\pi_n = \rho_0$ in graph $G$ with $\pi \circ \rho = \pi_0, \dots, \pi_n, \dots, \rho_n$ resulting in a $\pi_0 \text{-}\rho_n$ walk in $G$.

\paragraph{Tree decomposition and treelength.}
A \textbf{tree decomposition} of $G = (V,E)$ is an ordered pair $(T,\{V_t \mid t \in V(T)\})$ where $T$ is a tree and $V_t \subseteq V(G)$ is a subset associated with each vertex of $T$. The tree $T$ and the collection of subsets $\{V_t \mid t\in V(T)\}$, which are often referred to as \textbf{bags} of the tree decomposition, must satisfy the following properties.

\begin{itemize}
    \item Every vertex of $G$ is contained in at least one bag $V_t$ i.e.~$V(G) = \bigcup\limits_{t\in V(T)} V_t$. 
    \item For each edge $e$ of $G$, there is a bag $V_t$ which contains both endpoints of $e$.
    \item Let $t,t'$ and $t''$ be three vertices of $T$ such that $t'$ lies on the unique path from $t$ to $t''$ in $T$. Then, if a vertex $v$ of $G$ belongs to both $V_{t}$ and $V_{t''}$, it also belongs to $V_{t'}$.
\end{itemize}

The \textbf{treelength} of a graph $G = (V,E)$, denoted $\tl(G)$, is the minimum integer $j$ such that there exists a tree decomposition $(T,\{V_t \mid t \in V(T)\})$ where $d_G(u,v) \leq j$ for all $u,v \in V_t$ for $t \in V(T)$. That is, any two vertices in the same bag of the tree decomposition are at distance at most $j$. 

For example, the treelength of a complete graph is one, as one can place all vertices in a single bag. More generally, the treelength of a chordal graph is one, and of a $k$-chordal graph is $O(k)$.

\paragraph{The layering tree.}
A key tool that our algorithm uses is the notion of layering tree as used by Chepoi and Dragan \cite{Chepoi2000}. A similar notion has appeared in the literature of graph theory, computational geometry, and topology under different names such as merge tree, join tree or split tree (see these references: \cite{CSA2003,WWW2014}).

Let $G= (V, E)$ be a connected graph, and let $s\in V$ be an arbitrary vertex. 
We define Breath First Search (BFS) \textbf{layers} $L_i$ to be the set of vertices at shortest path distance $i$ from $s$, i.e.~$L_i = \{v \mid d_G(s,v) = i\}$.  In particular, $L_0 = \{s\}$.
We use notations, $L_{\leq k} := L_0\cup\ldots\cup L_k$, and $L_{\geq k} := L_k\cup\ldots\cup L_{n-1}$.

For any $i\geq 0$, let $S_i = \{S_i^1, \ldots, S_i^{s_i}\}$ be the set of connected components of $G\setminus L_{\leq i-1}$, noting that $L_{-1} = \emptyset$.  In turn, for any $1\leq j\leq s_i$, let $P_i^j = S_i^j\cap L_i$, the set of vertices in $S_i^j$ that belong to $L_i$.  We call the elements of ${\cal P}_i = \{P_i^1, \ldots, P_i^{s_i}\}$ the \textbf{parts} of $G$ at layer $i$.  It follows that these parts partition $L_i$.  Moreover, the set of all parts (across all layers), ${\cal P} = \bigcup_{0\leq i}{{\cal P}_i}$, partitions $V(G)$.

We define the \textbf{layering tree} ${\cal T} = ({\cal P}, {\cal E})$ of $G$ with respect to $s$ to be the graph whose vertices are the \emph{parts} of $G$.  For two parts $P, P'\in {\cal P}$, $(P, P')\in {\cal E}$ if and only if there exists $u\in P$ and $u'\in P'$ such that $(u,u')\in E(G)$. 
\begin{figure}
    \centering
    \includegraphics[scale = 0.35]{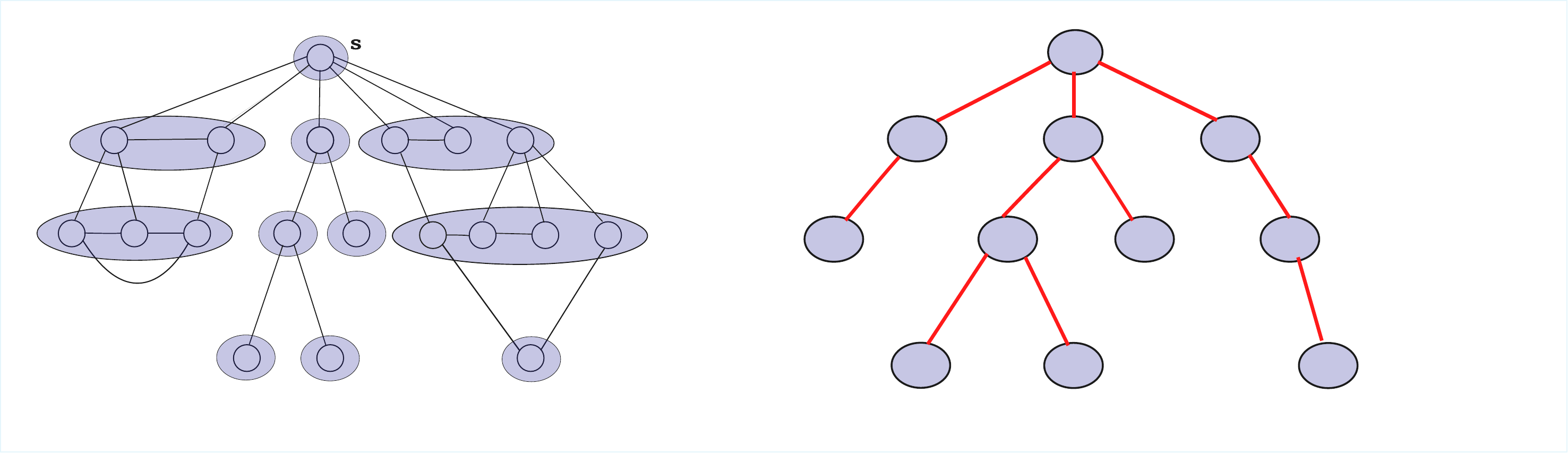}
    \caption{Left: the parts of a graph, shaded in purple, based on the BFS layering done with respect to $s$; Right: the layering tree whose vertices correspond to the parts.}
    \label{fig:LayeringTreeFig}
\end{figure}

In this paper, we refer to parts both as vertices of the layering tree and as subsets of $V(G)$, but it should be clear from the context which one we refer to. Our algorithm needs to consider the subtrees induced by subsets of ${\cal P}$.  We use ${\cal T}_k$ to refer to the induced subtree ${\cal T}[\bigcup_{i \leq k-1} \mathcal{P}_i]$ i.e.~the subtree of $\cal{T}$ induced by the parts at layers $0, 1, \ldots, k-1$, and \emph{not} $k$.

Further we define the \textbf{length} of a layering tree, denoted $\ell({\cal T})$, to be the maximum diameter of any of its parts, i.e.~$\ell({\cal T}) = \max_{P\in{\cal P}}(\max_{u,v\in P}(d_G(u,v)))$.  By a simple counting argument, the size of each part of a layering tree can be bounded.
\begin{lemma}
\label{lem:max_part_size}
Let $G=(V,E)$ be a graph with maximum degree $\Delta$, and let ${\cal T} = ({\cal P}, {\cal E})$ be any of its layering trees.
The size of each part in ${\cal P}$ is at most $\frac{\Delta^{\ell(\T)+1}-1}{\Delta -1} \leq \Delta^{\ell(\T)+1}$
\end{lemma}

Our definition of layering tree is equivalent to the definition used by Dourisbourne and Gavoille~\cite{journals/dm/DourisboureG07}--in their definition, two vertices in $L_i$ belong to the same part if and only if there is a path between them in $G[L_{\geq i}]$.  
They show, in Theorem 8 of their paper~\cite{journals/dm/DourisboureG07}, the following bound on the length of a layering tree of $G$ with respect to its treelength.

\begin{lemma}[Dourisbourne and  Gavoille~\cite{journals/dm/DourisboureG07}, Theorem 8]
\label{lem:diam_layering_tree}
    Let $G=(V,E)$ be a connected graph with treelength $\tl(G)$, and let ${\cal T} = ({\cal P}, {\cal E})$ be its layering tree (with respect to an arbitrary vertex $s$). Then $\ell(\T) \leq 3\tl(G)$.
\end{lemma}

Chepoi and Dragan~\cite{Chepoi2000} show that the layering tree of a graph can be computed in linear time.  In this paper, we focus on query complexity rather than time complexity. However, it is not hard to show that the asymptotic time complexity of our algorithm is the same as its query complexity once the maximum degree and treelength are treated as constants.

\section{Reconstruction Algorithm}

Let $G=(V,E)$ be a hidden graph with a visible vertex set of cardinality $n$ and with a shortest path distance oracle.  Our algorithm first selects an arbitrary vertex $s\in V(G)$ and uses $n-1$ queries to find $d_G(s,v)$ for all $v\in V(G) \setminus \{s\}$.  Based on this information, our algorithm can build the BFS layers $L_0, L_1, \ldots, L_{n-1}$.

Next, our algorithm reconstructs $G$ layer by layer.  Specifically, for each vertex $v\in L_i$ our algorithm uses SP queries to find the neighbors of $v$ in $L_i \cup L_{i-1}$.

Let $G_i = G[L_{\leq i-1}]$ be the subgraph of $G$ induced by the first $i$ layers.  
Also, recall that $\T$ is the layering tree rooted at $s$, and let $\ell$ be an upper bound on its length, i.e.~$\ell\geq \ell(\T)$. Similarly, let $\T_k$ be the subtree of $\T$ induced by the first $k = i - \ell -2$ layers of $\T$, the layering tree. 
Our algorithm alternates between extending $G_i$ to $G_{i+1}$ and $\T_k$ to $\T_{k+1}$.  Specifically, given $G_i$, we construct $\T_k$ with no shortest path queries.  Then, given $\T_k$ and $G_i$, we construct $G_{i+1}$ with $|L_i|\cdot \Delta^{O(\ell)}\cdot\log n$ shortest path queries.  Our latter algorithm is composed of a binary search type algorithm in $\T_k$ and an exhaustive search within layers $k$ to $i$, which is similar in spirit to Bastide and Groenland~\cite{Bastide2024a} reconstruction for $k$-chordal graphs. 

\subsection{Extending the layering tree}
We first show that given $G_i$, one can reconstruct the first $k = i - \ell - 2$ layers of $\T$, denoted $\T_{k}$. By the definition of a layering tree, if there exists a path in $G[L_{\geq k}]$ between two vertices $u, v \in L_k$, then $u$ and $v$ belong to the same part at layer $k$; in particular, if there exists a $u\text{-}v$ path within the next $O(\ell)$ layers of $G$, then $u$ and $v$ belong to the same part at layer $k$. The following lemma is a key structural insight for our construction that shows the converse is also true, i.e.,~showing that if two vertices belong to the same part at layer $k$ then there is a path between them within the next $O(\ell)$ layers of $G$. Thus, this lemma implies that to construct $\T_{k}$, we only need to know the first $k + O(\ell)$ layers of $G$.

\begin{lemma}
\label{lem:ell_Search}
    Let $G=(V,E)$, let $\T= ({\cal P}, \E)$ be a layering tree of $G$, and let $\ell \geq \ell(\T)$.  Suppose $u,v\in V(G)$ belong to the same part $P$ of $\T$ at layer $k$.  Then, there is a path between $u$ and $v$ in $G$ whose vertices are in $L_{\leq k+\ell+1} \cap L_{\geq k}$, i.e.~$u$ and $v$ are connected in $G\setminus(L_{\leq k-1}\cup L_{\geq k+\ell+2})$. 
\end{lemma} 

\begin{proof}
    Suppose there is no $u\text{-}v$ path in $G\setminus(L_{\leq k-1}\cup L_{\geq k+\ell+2})$. Since $u,v \in P$, by definition there exists a $u\text{-}v$ path in $G \setminus L_{\leq k-1}$. Let $\pi$ be a shortest such $u\text{-}v$ path. Then $\pi$ has a vertex $t$ in $L_{k+\ell+2}$. 
    Let $\mu = \lfloor\frac{\ell}{2}\rfloor$, and
    let $x$ be the last vertex in $L_{k+\mu}$ on the subpath $\pi[u,t]$ and $y$ be the first vertex in $L_{k+\mu}$ on the subpath $\pi[t,v]$.  
    Note that $x$ and $y$ are both in $ L_{k+\mu}$ and there exists a path $\pi[x,y]$ in $G \setminus L_{\leq k+\mu - 1}$ that is composed of $\pi[x,t]$ and $\pi[t,y]$, both in $G \setminus L_{\leq k+ \mu - 1}$. Therefore, $x$ and $y$ belong to the same part of $\T$. It follows that $d_G(x,y) \leq \ell(\T) \leq \ell$, as the diameter of each part is at most $\ell$. 

    By definition of $\pi$, the path $F:= \pi[x,t] \circ \pi[t,y]$ is a shortest $x\text{-}y$ path whose vertices are in $L_{\geq k}$. Thus, any $x\text{-}y$ path either has length at least the length of $F$ or contains a vertex of layer $L_{k-1}$. That is, $d_G(x,y) \geq \min\{2(\ell+2 - \mu), 2(\mu+1)\} \geq \ell + 1$, a contradiction.
\end{proof}

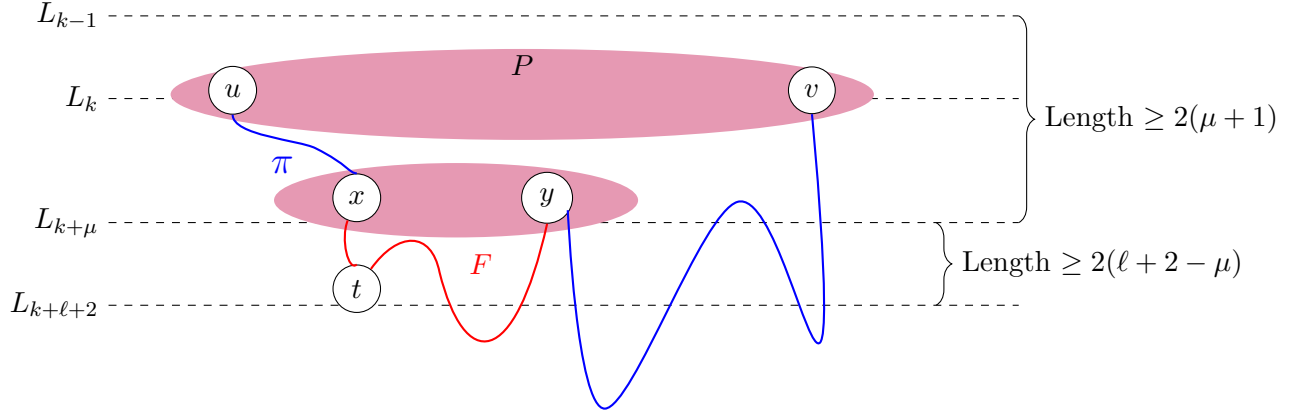
\begin{figure}
    \centering
    \begin{tikzpicture}[scale=1.1]

\draw[dashed] (-1,3.5) -- (10,3.5);
\draw[dashed] (-1,2.5) -- (10,2.5);
\draw[dashed] (-1,1) -- (10,1);
\draw[dashed] (-1,0) -- (10,0);

\node[left] at (-1,3.5) {$L_{k-1}$};
\node[left] at (-1,2.5) {$L_k$};
\node[left] at (-1,1) {$L_{k+\mu}$};
\node[left] at (-1,0) {$L_{k+\ell+2}$};

\fill[purple!40, rounded corners=12pt] (4.0,2.55) ellipse (4.25 and 0.55);
\fill[purple!40, rounded corners=12pt] (3.2,1.27) ellipse (2.2 and 0.45);

\node[circle,draw,fill=white] (u) at (0.5,2.6) {$u$};
\node[circle,draw,fill=white] (x) at (2.0,1.3) {$x$};
\node[circle,draw,fill=white] (y) at (4.3,1.3) {$y$};
\node[circle,draw,fill=white] (v) at (7.5,2.6) {$v$};
\node[circle,draw,fill=white] (t) at (2.0,0.2) {$t$};


\node at (4.0,2.9) {$P$};


\draw[thick, blue]
(u) .. controls (0.5,2.05) and (1.3,2.0) .. (1.5,1.9)
    .. controls (1.9,1.7) and (1.9,1.6) .. (x.north);
\node at (1.1,1.7) {\Large $\mathcolor{blue}\pi$};

\draw[thick, blue]
plot[smooth, tension=0.6] coordinates {
    (4.55,1.15)
    (5.0,-1.25)
    (6.6,1.25)
    (7.6,-0.45)
    (v.south)
};

\draw[thick, red] (x) .. controls (1.8,0.8) and (1.9,0.4) .. (t.north);


    \draw[thick, red]
(t)
  .. controls (2.5,0.9) and (2.9,0.9) .. (3.0,0.4)
  .. controls (3.3,-0.8) and (3.9,-0.8) .. (y.south);
\node at (3.5,0.5) {$\mathcolor{red}F$};
    
\draw[decorate,decoration={brace,amplitude=6pt}]
(10,3.5) -- (10,1)
node[midway,right=6pt] {Length $\ge 2(\mu+1)$};

\draw[decorate,decoration={brace,amplitude=6pt}]
(9,1) -- (9,0)
node[midway,right=6pt] {Length $\ge 2(\ell+2-\mu)$};

\end{tikzpicture}
    \caption{A shortest $u\text{-}v$ path $\pi$ in $G \setminus L_{\leq k-1}$ and its $x \text{-} y$ subpath $F$ (in red) in the proof of \Cref{lem:ell_Search}. Parts are represented by the shaded ellipses. }
    
\end{figure}

We use \lemref{lem:ell_Search} to construct $\T_k$ from $G_i$. Note that to construct $\T_k$, one only needs to be able to check if a pair of vertices $u, v\in L_k$ are connected in the graph $G\setminus L_{\leq k-1}$.  By \lemref{lem:ell_Search}, to check this connectivity we only need to consider $G_i$.
\begin{corollary}
\label{lem:T_from_H}
    Let $G=(V,E)$, and let $\T= ({\cal P}, \E)$ be a layering tree of $G$. For any $i\geq \ell(\T) + 2$ and $k \leq i-\ell(\T)-2$, there is an algorithm to compute $\T_k$ from $G_i$ (and no more information about $G$) with no shortest path queries.
\end{corollary}

\subsection{Extending graph reconstruction}

Next, we show we can reconstruct $G_{i+1}$ given $G_{i}$ and $\T_k$ for $k = i - \ell -2$. That is, we can find all the edges that are in $G_{i+1}$ but not in $G_{i}$, namely the edges with both endpoints in $L_{i+1}$ or one endpoint in $L_{i}$ and one in $L_{i+1}$.

Our algorithm takes two steps to achieve this reconstruction.  First, for each vertex $v\in L_{i}$, it finds the connected component of $v$ in $G\setminus L_{\leq k-1}$; more specifically, the part at layer $k$ that is the ancestor (in the partially constructed layering tree $\mathcal{T}_k$) of the part that contains $v$.  This step requires $\Delta^{O(\ell)}\cdot \log n$ shortest path queries.  Once this connected component is discovered, our algorithm exhaustively searches for all neighbors of $v$ within this connected component at layers $L_{i}$ and $L_{i-1}$ to find all its neighbors. This step takes $\Delta^{O(\ell)}$ queries per vertex in $L_i$.

To prepare for a more formal description of our method, we provide a few definitions:
\begin{itemize}
    \item For a vertex $v \in V(G)$, we define $\textbf{part}(v)$ to be the part of the layering tree $\T$ that contains $v$. Note that for any vertex $v \in L_{\leq k-1}$, the vertex of $\mathcal{T}_k$ which contains $v$ is equivalent to $\part(v)$, so we will use $\part(v)$ for such vertices $v$ despite possibly not having reconstructed $\mathcal{T}$. 

    \item For a vertex $v\in L_i$ and an integer $k\in \{0, 1, \ldots, i\}$, we define $\textbf{anc}(v,k)$ to be the part of $\T$ that is the ancestor of $\part(v)$ at depth $k$ e.g.~$\anc(v, i) = \part(v)$, and $\anc(v, 0) = \part(s) = \{s\}$, where $s$ is the root of the BFS tree.

    \item 
        For a part $P$ of $\T$, we define $\textbf{comp}(P)$ to be the subset of vertices contained in $P$ and all of the descendants of $P$ in $\T$. That is, the subset of $V(G)$ that is contained in the subtree of $\T$ rooted at $P$. Note, if $P$ is a part at layer $k$ then $\comp(P)$ is the connected component of $G\setminus L_{\leq k-1}$ that contains the vertices in $P$.
\end{itemize}

\paragraph{Logarithmic search for the connected component.}
Let $v\in L_i$, and let $k\in \{0,1, \dots, i\}$.  The first step of our algorithm is a logarithmic search to find $\anc(v, k)$.  To that end, we use the following classic result on vertex separators of trees. A \textbf{$(1/2)$-cut vertex} or \textbf{centroid} of graph $G$ is a vertex whose removal results in connected components whose size is at most $\frac{1}{2}|V(G)|$. 

\begin{lemma}
\label{lem:half_seperater}
Every tree has a centroid. (see \cite[Ch.~1]{Diestel17})
\end{lemma}

\begin{lemma}
\label{lem:extend_comp_levels_one_vertex}
    Let $G=(V,E)$, let $\T= ({\cal P}, \E)$ be a layering tree of $G$, and let $\ell \geq \ell(\T)$.  Also, let $i\geq k\geq 0$, and suppose $\T_k$ and $G_{i}$ are known.
    For any $x\in L_i$, one can find $\anc(x, k)$ with $O(\Delta^{\ell+2}\cdot \log n)$ queries. 
\end{lemma}
\begin{proof}
Since $\T_k$ is a tree, by \lemref{lem:half_seperater} it contains a $(1/2)$-cut vertex. The proof constructs a sequence of subtrees of $\T_k$, namely $(\T_k^j)^{f}_{j=0}$ and a sequence of $(1/2)$-cut vertices $(P_j)^{f-1}_{j=0}$ for these subtrees such that (1) $f \leq \lceil \log n\rceil$, and (2) $\T_k^f$ has only one vertex: $\mathrm{anc}(x, k)$. Initially, $\T_k^0 = \T_k$. The tree $\T_k^{j+1}$ is defined as the connected component of $T_k^j \setminus \{P_j\}$ that contains $\mathrm{anc}(x, k)$.

Since $P_j \in V(\T_k^j)$ is a separator, $\T_k^{j+1}$ is the subtree of $\T_k^{j}$ that contains the neighbor of $P_j$ which contains the vertex closest to $x$ among all vertices in $N_G(P_j)$, the set of all vertices of $V(G) \setminus P_j$ that have a neighbor in $P_j$.  Therefore, to find $\T_k^{j+1}$ it suffices to query the distance from $x$ to all the vertices in $N_G(P_j)$. Note that $\anc(x,k)$ is a leaf of $\mathcal{T}_k$ and thus we can always avoid choosing it, or any other leaf, as a vertex in the sequence $(P_j)^{f-1}_{j=0}$; this also guarantees that $N_G(P_j) \subseteq G[V(\mathcal{T}_k)]$ for all $P_j$ in the sequence. 

By \lemref{lem:max_part_size} there are at most $\Delta^{\ell+1}$ vertices in $P_j$. Since each vertex in $N_G(P_j)$ has at least one neighbor in $P_j$ and $G$ has maximum degree $\Delta$, we have that \[
|N_G(P_j)|\leq \Delta\cdot |P_j| = \Delta^{\ell+2}
\]
Since $\T_k$ has at most $n$ parts and we eliminate at least half of the parts at each step, we know that this search takes at most $\lceil\log_2 n\rceil$ steps, therefore the total query complexity is
\(
\Delta^{\ell+2} \cdot\lceil\log_2 n\rceil = O(\Delta^{\ell+2}\cdot \log n)
\)
\end{proof}
Based on the lemma above, we can compute the connected component of each $v\in L_i$ in $G\setminus L_{\leq k-1}$. 
\begin{lemma}
\label{lem:extend_comp_levels}
    Let $G=(V,E)$, let $\T= ({\cal P}, \E)$ be a layering tree of $G$ and let $\ell \geq \ell(\T)$.  Also, let $i\geq k\geq 0$, and suppose $\T_k$ and $G_{i}$ are known.
There is an algorithm to compute $\comp(P)\cap L_{\leq i}$ for all $P$ at depth $k$ with $O(|L_i|\cdot \Delta^{\ell+2}\cdot \log n)$ queries.   
   
\end{lemma}
\begin{proof}
To compute $\comp(P)\cap L_{\leq i}$ for all $P$ in $\T$ at depth $k$, it suffices to find $\anc(u, k)$ for all $u\in L_{k}\cup\ldots\cup L_i$. Since we have $G_i$, which is equal to $G[L_{\leq i-1}]$, for each $u\in L_{k}\cup\ldots\cup L_{i-1}$ we have its shortest path to $L_k$ as part of the BFS tree, therefore we can find $\anc(u, k)$ by checking which part at depth $k$ is connected to $u$ in $G\setminus L_{\leq k-1}$.  We do not need any shortest path queries for this step as $G_i$ is given.  For vertices $u\in L_i$, we use \lemref{lem:extend_comp_levels_one_vertex} to find $\anc(u, k)$ with $O(\Delta^{\ell+2}\cdot \log n)$ queries. Since there are $|L_i|$ vertices at layer $i$, the total number of queries is as stated in the lemma.
\end{proof}

\paragraph{Brute force search to find neighbors.} 
So far, we have computed $\comp(P)\cap L_{\leq i}$ for all $P$ at depth $k$.  We bound the size of each of these vertex sets.
\begin{lemma}
\label{lem:comp_set_bound}
    Let $G=(V,E)$, let $\T= ({\mathcal{P}}, \E)$ be a layering tree of $G$, and let $\ell\geq \ell(\T)$. Also, let $P\in\mathcal{P}$ be a part of $\T$ with depth $k$.  Also, let $i\geq k$. 
    Then 
    \(
    |\comp(P)\cap L_{\leq i}|
    \leq \Delta^{\ell+i-k+2}.
    \)
   
\end{lemma}
\begin{proof}
    Each vertex in $\comp(P)\cap L_{\leq i}$ is at distance at most $i-k$ from a vertex in $P$.  Since the graph has maximum degree $\Delta$, it follows that there are at most $|P|\cdot \Delta^{i-k+1}$ such vertices.  By \lemref{lem:max_part_size}, $|P|\leq \Delta^{\ell+1}$, which implies the bound of the lemma.
\end{proof}

Putting together \lemref{lem:extend_comp_levels} and \lemref{lem:comp_set_bound}, our algorithm constructs $G_{i+1}$ from $G_i$ and $\T_k$.

\begin{lemma}
\label{lem:H_from_T}
    Let $G=(V,E)$, let $\T= ({\mathcal{P}}, \E)$ be a layering tree of $G$, and let $\ell\geq \ell(\T)$. 
    Also, let $k\geq 0$ and $i = k + \ell+2$. Suppose $G_i$ and $\T_k$ are given.  
    There is an algorithm to compute $G_{i+1}$ with $O(|L_i|(\Delta^{\ell+2}\log n + \Delta^{4\ell+8}))$ queries.
\end{lemma}
\begin{proof}

By \lemref{lem:extend_comp_levels}, $\comp(P)\cap L_{\leq i}$ can be computed for all $P$ in $\T$ at depth $k$  with $O(|L_i|\cdot \Delta^{\ell+2}\cdot \log n)$ queries. To compute edges between $L_i$ and $L_{i-1}$ and within $L_i$, it suffices to check whether every vertex $v\in L_i$ is adjacent to vertices in the same connected component as $v$ in $G\setminus L_{\leq k-1}$.  To do that we check the set $\comp(\anc(v, k))\cap L_{\leq i}$ for possible neighbors of $v$. By \lemref{lem:comp_set_bound}, these sets are bounded by $\Delta^{\ell+i - k+2} = \Delta^{2\ell+4}$.  Overall, we have at most $|L_i|\cdot \Delta^{4\ell+8}$ queries of the latter type. Hence, the total number of queries is $O(|L_i|(\Delta^{\ell+2}\log n + \Delta^{4\ell+8}))$
\end{proof}

\subsection{Proof of the main result} 

Finally, we are ready to prove our main result.

\begin{proof}[Proof of \Cref{thm:reconstruction_algo}]
    First, we select an arbitrary vertex $s$ in $V(G)$ and we query the distance from $s$ to every $v \in V(G)\setminus\{s\}$ to build the layers $L_0, \ldots, L_{n-1}$; this takes $n-1$ queries. Let $\T$ be the layering tree of $G$ rooted at $\{s\}$ (note, we have not yet computed $\T$).  By \lemref{lem:diam_layering_tree}, we know $\ell(\T) \leq 3\tl(G)$. We set $\ell := 3\tau \geq 3\tl(G)$, where $\tau$ is from the theorem statement. Hence $\ell \geq \ell(\T)$.
    
    We compute $G_{\ell+2}$ by querying every pair of vertices in $L_{\leq \ell+1}$, taking at most $\Delta^{2\ell+4}$ queries. 

    For $i\geq \ell+2$, we compute $G_{i+1}$ from $G_i$ in two steps.  First, we compute $\T_k$ for $k = i -(\ell+2)$, using no queries, by the algorithm of \Cref{lem:T_from_H}.
    Then, by \lemref{lem:H_from_T}, we compute $G_{i+1}$ using at most $|L_i|(\Delta^{\ell+2}\cdot\log n + \Delta^{4\ell+8})$ queries. Summing over all remaining layers, this quantity is at most $n(\Delta^{\ell+2} \cdot \log n + \Delta^{4\ell+8})$. Therefore, we reconstruct $G$ with at most 
    \begin{align*}
        n-1&+\Delta^{2\ell+4} + n(\Delta^{\ell+2} \cdot \log n + \Delta^{4\ell+8}) \\ &= O(\Delta^{\ell+2}\cdot n \log n) = O(\Delta^{3\tau+2}\cdot n \log n)    
    \end{align*}
    queries.
\end{proof}
\bibliography{main}
\bibliographystyle{alpha}
\end{document}